%% file: main.tex
\documentclass[conference]{IEEEtran}
\IEEEoverridecommandlockouts
\usepackage[T1]{fontenc}
\usepackage{graphicx}
\usepackage{cite}
\usepackage{subcaption}

\usepackage{overpic}
\usepackage{amssymb}
\usepackage{amsmath}
\usepackage{amsthm}
\usepackage{amssymb}
\usepackage{steinmetz}
\usepackage{array}
\usepackage{lipsum}  
\usepackage{url}
\usepackage[dvipsnames]{xcolor}
\usepackage{mathrsfs}
\usepackage{todonotes}

\newcommand{\bsy}[1]{\boldsymbol{#1}}
\newcommand{\derpar}[2]{\dfrac{\partial {#1}} {\partial {#2}}}

\theoremstyle{plain}

\newtheorem{lemma}{Lemma}

\newtheorem{corollary}{Corollary}
\newtheorem{remark}{Remark}

\newtheorem{assumption}{Assumption}

\newcommand{\vect}[1]{\mathbf{#1}}

\def\imagunit{\mathsf{j}} 

\makeatletter
\renewcommand\boxed[2][\fboxsep]{
  {%
    \setlength\fboxsep{#1}%
    \fbox{\m@th$\displaystyle#2$}
  }%
}
\makeatother

\usepackage{fontawesome}
\usepackage{calligra}

\DeclareMathAlphabet{\mathcalligra}{T1}{calligra}{m}{n}

\usepackage{pgfplots,pgfplotstable,relsize}
\pgfplotsset{compat=newest}
\usepackage{tikz}
\usetikzlibrary{shapes,shadows,arrows,intersections}

\tikzset{
small dot/.style={fill=black,circle,scale=0.02}
}

\begin{document}

\title{Cram{\'e}r-Rao Bounds for Near-Field Localization}


\author{
\IEEEauthorblockN{Andrea de Jesus Torres\IEEEauthorrefmark{1}, Antonio A. D'Amico\IEEEauthorrefmark{1}, Luca Sanguinetti\IEEEauthorrefmark{1},
Moe Z. Win\IEEEauthorrefmark{2}
\thanks{\newline \indent The research was supported by the MIT-UNIPI grant (VIII call) from MISTI Global Seed Funds in the framework of the MIT-Italy Program. L. Sanguinetti and A. A. D'Amico were also partially supported by the Italian Ministry of Education and Research (MIUR) in the framework of the CrossLab project (Departments of Excellence).}}
\IEEEauthorblockA{\IEEEauthorrefmark{1}\small{Dipartimento di Ingegneria dell'Informazione, University of Pisa, Italy}}
\IEEEauthorblockA{\IEEEauthorrefmark{2}\small{Massachusetts Institute of Technology, Cambridge, USA}\vspace{-0.3cm}}
}

\maketitle

\begin{abstract}
Multiple antenna arrays play a key role in wireless networks for communications but also localization and sensing. The use of large antenna arrays pushes towards a propagation regime in which the wavefront is no longer plane but spherical. This allows to infer the position and orientation of an arbitrary source from the received signal without the need of using multiple anchor nodes. To understand the fundamental limits of large antenna arrays for localization, this paper fusions wave propagation theory with estimation theory, and computes the Cram{\'e}r-Rao Bound (CRB) for the estimation of the three Cartesian coordinates of the source on the basis of the electromagnetic vector field, observed over a rectangular surface area. To simplify the analysis, we assume that the source is a dipole, whose center is located on the line perpendicular to the surface center, with an orientation a priori known. Numerical and asymptotic results are given to quantify the CRBs, and to gain insights into the effect of various system parameters on the ultimate estimation accuracy. It turns out that surfaces of practical size may guarantee a centimeter-level accuracy in the mmWave bands.

\end{abstract}
\smallskip

\section{Introduction}

The estimation accuracy of signal processing algorithms for positioning is fundamentally limited by the quality of the underlying measurements. For time-based measurements, high resolution and high accuracy can only be obtained when a large bandwidth is available. Improvements can be achieved by using multiple anchor nodes. Antenna arrays have thus far only played a marginal role in positioning since the small arrays of today's networks provide little benefit. With future networks, the situation may change significantly. Indeed, the 5G technology standard is envisioned to operate in bands up to 86\,GHz~\cite{Lee2018Spectrum5G}, while 6G research is already focusing on the so-called sub-terahertz (THz) bands, i.e., in the range 100 -- 300\,GHz. The small wavelengths of high-frequency signals make it practically possible to envision arrays with a very large number of finely tailorable antennas, as never seen before. The advent of large spatially-continuous electromagnetic surfaces interacting with wireless signals pushes even further this vision. Research in this direction is taking place under the names of Holographic MIMO~\cite{Huang2020,Pizzo2019a,dardari2020holographic}, large intelligent surfaces~\cite{Hu2018b}, and reconfigurable  intelligent surfaces~\cite{Basar2019a,DiRenzo2020}. All this opens new dimensions and brings new opportunities for communications but also for localization and sensing.

An unexplored and unintentional side-effect of using large arrays or surfaces combined with high carrier frequencies, is to push the electromagnetic propagation regime from the Fraunhofer far-field region towards the Fresnel near-field region~\cite{dardari2020holographic}. This opens the door to new signal processing algorithms that exploit the unique near-field properties to pinpoint the position of the source with high accuracy~\cite{Hu2018b,Dardari2020Spherical,Rusek2020SphericalLIS,AlegriaLISPositioning2019,JinLensArray2020}. In this context, the question arises of the ultimate accuracy that can be achieved in localization operations. This is important in order to provide benchmarks for evaluating the performance of actual estimators. Motivated by this, this paper starts from first electromagnetic principles and provides the vector field observation over a rectangular spatial region, as a function of the radiation vector at the source. This is then used to compute the CRB for its three Cartesian coordinates. To simplify the analysis, we consider a dipole, located on the line perpendicular to the surface center, with an orientation a priori known.


%
\input{SectionII}

\input{SectionIII}
\input{SectionIV}

\section{Conclusions}
Large antenna arrays and high frequencies push towards the near-field regime, which opens up opportunities for new signal processing algorithms for positioning. Motivated by the need of establishing ultimate bounds, we considered the electromagnetic field over a rectangular spatial region as a function of the radiation vector at the source. This was used to compute the CRB for the three-dimensional (3D) spatial location of a dipole, whose center is on the line perpendicular to  surface center and whose 3D orientation is a priori known. Numerical results showed that a centimeter-level accuracy can be achieved in the near-field of surfaces of practical size (i.e., in the range of a few meters) in the mmWave and sub-THz bands. Asymptotic expressions were also given in closed-form to show the scaling behaviors with respect to surface area and wavelength. 

The ultimate goal of positioning is to precisely estimate not only the 3D spatial location, but also the 3D orientation of the source. This requires the computation of the CRB with no a priori knowledge of the orientation, which is addressed in the extended journal version~\cite{damico2021}.

\bibliographystyle{IEEEtran}
\bibliography{refs}

\end{document}

%% file: SectionII.tex
 \begin{figure}
	
	\centering
	\begin{overpic}[width=1\columnwidth]{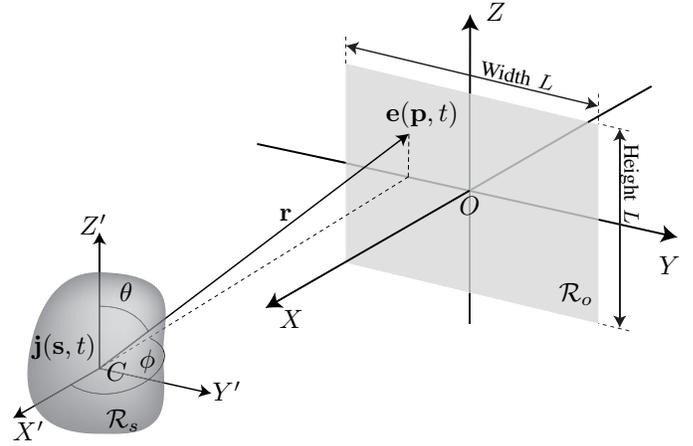}
	\put(67,33){$O$} 
	\put(40,16){$X$}
	\put(97,24){$Y$}
	\put(71,62){$Z$}
	
	\put(14,8){$C$}
	\put(0,-1){$ X'$}
	\put(30,5){$Y'$}
	\put(10,30){$Z'$}

	\put(56,47){${\mathbf{e}}({\mathbf{p}},t)$}
	
	\put(3,12){${\mathbf{j}}({\mathbf{s}},t)$}

	\put(14,1){$\mathcal R_s$}
	\put(82,20){$\mathcal R_o$}
	
	\put(40,32){$\mathbf{r}$}
	\put(16,20){$\theta$}
	\put(19,10){$\phi$}
	
	\put(70,54){\rotatebox{-13}{\footnotesize{Width $L$}}}
	
    \put(91.1,43.6){\rotatebox{-90}{\footnotesize{Height $L$}}}

	\end{overpic}	
\caption{Geometry of the considered system.}\vspace{-0.5cm}
	\label{fig:source_volume}
\end{figure}

\section{Signal model and problem formulation}

Consider the system depicted in Fig.~\ref{fig:source_volume} in which an electric current density ${\bf j}({\bf s}, t)$, inside a source region $\mathcal R_s$, generates an electric field ${\bf e}({\bf p}, t)$ at a generic location ${\bf p}$. We consider only monochromatic sources and fields of the form ${\bf j}(\vect{s},t) = \mathrm{Re}\left\{{\bf j}(\vect{s}) e^{\imagunit \omega t}\right\}$ and ${\bf e}(\vect{p}, t) = \mathrm{Re}\left\{{\bf e}(\vect{p}) e^{\imagunit \omega t}\right\}$, respectively. In this case, Maxwell's equations can be written only in terms of the current and field \textit{phasors}, ${\bf j}(\vect{s})$ and ${\bf e}(\vect{p})$ \cite[Ch. 1]{ChewBook}.

We denote by $C$ the \textit{centroid} of $\mathcal R_s$ and assume that the electric field ${\bf e}(\vect{p})$, produced by  ${\bf j}(\vect{s})$, is measured over a region $\mathcal R_o$ (\textit{observation region}) outside $\mathcal R_s$. The electromagnetic field propagates in a homogeneous and isotropic medium with neither obstacles nor reflecting surfaces. In other words, there is only a line-of-sight (LOS) link between $\mathcal R_s$ and $\mathcal R_o$. 

\subsection{Signal model}
The measured field is the sum of ${\bf e}(\vect{p})$ and a random noise field ${\bf n}(\vect{p})$, i.e.,
\begin{equation}
\label{SigNoise}
{\boldsymbol \xi} (\vect{p})={\bf e}(\vect{p})+{\bf n}(\vect{p})
\end{equation}
where ${\bf n}(\vect{p})$ is generated by electromagnetic sources outside $\mathcal R_s$. Consider a cartesian coordinate system with the origin in $C$, as shown in Fig. 1, and make the following assumptions.
\begin{assumption}
The observation volume is a square region parallel to the $Y'Z'$ coordinate plane. In particular, $\mathcal R_o=\big\{ (x',y',z'): x'=x'_o, |y'-y'_o| \le L/2, |z'-z'_o| \le L/2 \big\}$, where $(x'_o,y'_o,z'_o)$ are the cartesian coordinates of the center $O$ of $\mathcal R_o$ in the system $CX'Y'Z'$. 
\end{assumption}
The cartesian system $OXYZ$, shown in Fig. 1, is obtained by $CX'Y'Z'$ through a pure translation. The position of $C$ in the system $OXYZ$ is given by the three coordinates $(x_C,y_C,z_C)$. Accordingly, we have $x'=x-x_C$, $y'=y-y_C$ and $z'=z-z_C$.
\begin{assumption}
Let $r_o$ be the distance of $C$ from $\mathcal R_o$ and denote by $l_s$ the largest dimension of $\mathcal R_s$. We assume that $r_o \gg l_s$ and $r_o \gg 2 l_s^2 /\lambda$, where $\lambda=2 \pi c /\omega$ is the wavelength. These conditions define the so-called far-field or Fraunhofer radiation region~\cite[Ch. 14]{OrfanidisBook}.
\end{assumption}
In the Fraunhofer radiation region, the electric field ${\bf e}(\vect{p})$ can be approximated as~\cite[Ch. 14]{OrfanidisBook}
\begin{equation}
\label{E_P}
{\bf e}(\vect{p})= G(r) \left[R_{\theta}(\theta,\phi) {\bsy{\hat{\theta}}}+R_{\phi}(\theta,\phi) {\bsy{\hat{\phi}}} \right]
\end{equation}
where $(r,\theta,\phi)$ are the spherical coordinates of ${\bf p} \in \mathcal R_o$,
\begin{equation}
G(r)= -\imagunit k\eta \dfrac{e^{-\imagunit kr}}{4 \pi r}  
\end{equation}
is the scalar Green's function, $k=2 \pi / \lambda$ is the wavenumber, $\eta$ is the characteristic impedance of the medium while ${\bsy{\hat{\theta}}}$ and ${\bsy{\hat{\phi}}}$ are unit vectors along the $\theta$ and $\phi$ coordinate curves. Functions $R_{\theta}(\theta,\phi)$ and $R_{\phi}(\theta,\phi)$ are the components along the ${\bsy{\hat{\theta}}}$ and ${\bsy{\hat{\phi}}}$ directions, respectively, of the \textit{radiation vector} ${\bf R}(\theta,\phi)$. This is related to the source current distribution by the following equation~\cite{OrfanidisBook}
\begin{equation}
\label{RadVect}
{\bf R}(\theta,\phi)=\int_{\mathcal R_s} {\bf j}({\bf s}) e^{\imagunit {\bf k}(\theta,\phi) \cdot {\bf s}} d {\bf s}
\end{equation}
where ${\bf k}(\theta,\phi)=k \hat{\bf r}$ is the \textit{wavenumber vector}, $\hat{\bf r}$ is the unit vector in the radial direction, and ${\bf k}(\theta,\phi) \cdot {{\bf s}}$ is the dot product between ${\bf k}(\theta,\phi)$ and $\bf s$. From~\eqref{RadVect}, it follows that the electric field ${\bf e}(\vect{p})$ depends on the current distribution ${\bf j}(\vect{s})$ through $R_{\theta}(\theta,\phi)$ and $R_{\phi}(\theta,\phi) $.

Denote by $\xi_x(\vect{p})$, $\xi_y(\vect{p})$ and $\xi_z(\vect{p})$, the cartesian components of ${\boldsymbol \xi} (\vect{p})$ along the $\hat{\bf x}$, $\hat{\bf y}$ and $\hat{\bf z}$ directions, respectively. From \eqref{SigNoise}, we have
\begin{equation}
\label{SigNoiseX}
\xi_x(\vect{p})= e_x(\vect{p})+n_x(\vect{p})
\end{equation}
\begin{equation}
\label{SigNoiseY}
\xi_y(\vect{p})= e_y(\vect{p})+n_y(\vect{p})
\end{equation}
\begin{equation}
\label{SigNoiseZ}
\xi_z(\vect{p})= e_z(\vect{p})+n_z(\vect{p})
\end{equation}
where
\begin{equation}
\label{ }
\nonumber
e_u(\vect{p})={\bf e}(\vect{p}) \cdot \hat{\bf u}
\end{equation}
with $u \in \{x,y,z\}$. By using \eqref{E_P} we get
\begin{align}
\label{ex}
e_x(\vect{p})&=G(r) \left[R_{\theta}(\theta,\phi) \cos \theta \cos \phi-R_{\phi}(\theta,\phi) \sin \phi \right]\\
\label{ey}
e_y(\vect{p})&=G(r) \left[R_{\theta}(\theta,\phi) \cos \theta \sin \phi+R_{\phi}(\theta,\phi) \cos \phi \right]\\
\label{ez}
e_z(\vect{p})&=-G(r) R_{\theta}(\theta,\phi) \sin \theta.
\end{align}
A statistical model for the random field ${\bf n}(\vect{p})$ is needed. A common assumption (e.g., \cite{Jensen2008channel}--\hspace{1sp}\cite{Gruber2008EIT}) is to model ${\bf n}(\vect{p})$ as a spatially uncorrelated zero-mean complex Gaussian process with correlation function 
\begin{equation}
\label{ }
\mathrm{E} \left\{{\bf n}(\vect{p}) {\bf n}^{\dagger}(\vect{p}') \right\}=\sigma^2 {\bf I} \delta(\vect{p}-\vect{p}')
\end{equation}
where ${\bf I}$ is the identity matrix, $\delta(\cdot)$ is the Dirac's delta function, and $\sigma^2$ is measured in $\mathrm V ^2$, where $\mathrm V$ indicates volts \cite{Gruber2008EIT}.

\subsection{Problem formulation}
We aim at computing the Cram\'er-Rao bound (CRB) for the estimation of the position of $C$ based on the noisy observations ${\boldsymbol \xi} (\vect{p})$. As noticed earlier, this requires some information about the current distribution inside $\mathcal R_s$. Denote by ${\bf u}=(x_C,y_C,z_C)$ the vector collecting the \emph{unknown} coordinates of $C$ with respect to the cartesian system $OXYZ$. The CRB for the estimation of the $i$th entry of ${\bf u}$ is \cite{Kay1993a}
\begin{equation}
\label{eq:CRBs}	
\mathrm{CRB}(u_i) = \left[\mathbf{F}^{-1}\right]_{ii}
\end{equation}
where ${\bf F}$ is the Fisher's Information Matrix (FIM). The latter is a $3\times 3$ hermitian matrix, whose elements are computed as:
\begin{equation}
\label{FIM}
\begin{split}
[\mathbf{F}]_{m,n}=\dfrac{1}{\sigma^2}\int\limits_{-\frac{L}{2}}^{\frac{L}{2}} \int\limits_{-\frac{L}{2}}^{\frac{L}{2}}&\left(\derpar{e_x}{u_m} \derpar{e_x^\ast}{u_n}+\right.\\
&\hphantom{i}\left.\derpar{e_y}{u_m} \derpar{e_y^\ast}{u_n} +\derpar{e_z}{u_m} \derpar{e_z^\ast}{u_n}\right)dy dz
\end{split}
\end{equation}
where the integration is performed over the observation region $\mathcal R_o$. For notational simplicity, in \eqref{FIM} the dependence of $e_x$, $e_y$ and $e_z$ on $\bf p$ has been omitted.  
\begin{remark}
A different approach for estimating the position of $C$ is to make use of a {\textrm scalar}, instead of vectorial, field. For example, one could use only one of the three components of ${\bf e}(\vect{p})$. This may simplify the analysis but would result in lower performance, i.e., a larger CRB. An alternative approach is to consider a scalar field that is related to the component of the Poynting vector perpendicular to each point of the planar region. This component is proportional to $ \|{\bf e}(\vect{p})\|^2 \sin \theta \cos \phi$, and the associated scalar field is
\begin{align}\notag
E& \triangleq e^{-\imagunit kr} \sqrt{ \|{\bf e}(\vect{p})\|^2 \sin \theta \cos \phi} \\&=k\eta \dfrac{e^{-\imagunit kr} \sqrt{x_C}}{4 \pi r^{3/2}}  \sqrt{R^2_{\theta}(\theta,\phi)+ R^2_{\phi}(\theta,\phi)}.\label{PVScalarField}
\end{align}
In the case of an isotropic radiating source, $R^2_{\theta}(\theta,\phi)+ R^2_{\phi}(\theta,\phi)$ is independent of $\theta$ and $\phi$, and thus ~\eqref{PVScalarField} reduces to the scalar model considered in~\cite[Eq. (2)]{Hu2018b}. We stress that this scalar model represents a specific case, which is not valid in general. We will use~\eqref{PVScalarField} in the numerical analysis for comparisons.
\end{remark}

%% file: SectionIII.tex
\section{CRB computation  with a priori information about the current distribution}
\label{sec:est3D}
To evaluate and quantify the CRB, we assume that the current source is a dipole of length $l_s$, as shown in Fig.~\ref{fig:holo-surface}, and make the following assumption. 
\begin{assumption}\label{assumption-1}
The dipole is oriented along the $z-$axis and the orientation is known. 
\end{assumption}
Assumption~\ref{assumption-1} implies that \textit{full information} about the source current distribution is available. In this case, we have (e.g.~\cite[Ch.\ 4]{Balanis}):
\begin{equation}
\label{F_teta}
R_{\theta}(\theta,\phi)=l_s I_{in} \sin \theta \quad \quad R_{\phi}(\theta,\phi)=0
\end{equation}
where $I_{in}$ is the uniform current level in the dipole. 
Plugging \eqref{F_teta} into \eqref{ex}--\eqref{ez} yields
\begin{align}
\label{ExCaseStudy}
e_x&= \phantom{-}\imagunit \chi \dfrac{e^{-\imagunit kr}}{r} \sin \theta \cos \theta \cos \phi\\ 
\label{EyCaseStudy}
e_y&= -\imagunit \chi \dfrac{e^{-\imagunit kr}}{r} \sin \theta \cos \theta \sin \phi\\ 
\label{EzCaseStudy}
e_z&=\phantom{-}\imagunit \chi \dfrac{e^{-\imagunit kr}}{r} \sin^2 \theta 
\end{align}
where $\chi= \dfrac{\eta I_{in}}{2} \dfrac{l_s}{\lambda}$
is measured in volts. The dependence of $e_x$, $e_y$ and $e_z$, on $(x_C,y_C,z_C)$ is hidden in $(r,\theta,\phi)$. Indeed, we have
\begin{align}
\label{SferCart1}
r&=\sqrt{x^2_C+(y-y_C)^2+(z-z_C)^2}\\
\label{SferCart2}
\cos \theta &=\dfrac{z-z_C}{r}\\
\label{SferCart3}
\tan \phi &=-\dfrac{y-y_C}{x_C}
\end{align}
from which it follows that
\begin{align}
\label{}
\sin \theta \cos \theta \cos \phi&=\dfrac{x_C(z-z_C)}{r^2}\\
\sin \theta \cos \theta \sin \phi&=\dfrac{(y-y_C)(z-z_C)}{r^2}\\
\sin^2 \theta&=1-\dfrac{(z-z_C)^2}{r^2}.
\end{align}
By using the above identities into \eqref{ExCaseStudy}--\eqref{EzCaseStudy} yields
\begin{align}
\label{ExCaseStudy.2}
e_x&= \phantom{-}\imagunit \chi e^{-\imagunit kr} \dfrac{x_C(z-z_C)}{r^3}\\
\label{EyCaseStudy.2}
e_y&= -\imagunit \chi e^{-\imagunit kr} \dfrac{(y-y_C)(z-z_C)}{r^3}\\
\label{EzCaseStudy.2}
e_z&=\phantom{-} \imagunit \chi \dfrac{e^{-\imagunit kr}}{r}\left[1-\dfrac{(z-z_C)^2}{r^2}\right].
\end{align}
The computation of the Fisher's information matrix through \eqref{FIM} requires the derivatives of $e_x$, $e_y$ and $e_z$ with respect to $x_C$, $y_C$ and $z_C$. These can be obtained from \eqref{ExCaseStudy.2}--\eqref{EzCaseStudy.2} after lengthy but standard calculations, not reported here for space limitations. They are provided in the extended version of this paper~\cite[App. A]{damico2021}.
\begin{figure}[t!]
\centering
    \begin{overpic}[width=1\columnwidth]{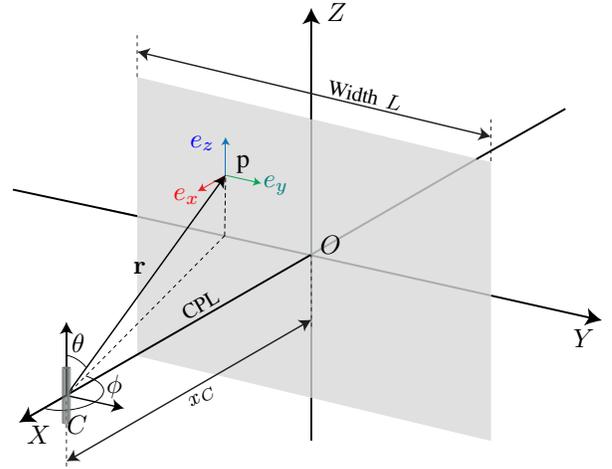}
    \put(50,38){$O$} 
    \put(88,24){$Y$}
    \put(51,73){$Z$}
    
    \put(12,11){$C$}
    \put(6,9.5){$X$}

    \put(37.5,51){p}
    \put(28,46){$\color{red}{e_x}$}
    \put(41.5,48){$\color{teal}{e_y}$}
    \put(30.5,54){$\color{blue}{e_z}$}
    
    \put(22,35){$\mathbf{r}$}
    \put(18,17){$\phi$}
    \put(12.7,23.2){$\theta$}
    
    \put(51,62){\rotatebox{-13}{\footnotesize{Width $L$}}}
    \put(29,28){\rotatebox{30}{\footnotesize\text{CPL}}}
    \put(30,15){\rotatebox{30}{\footnotesize\text{$ x_C$}}}

    \end{overpic}
    \caption{Source model and CPL assumption.}\vspace{-0.7cm}
    \label{fig:holo-surface}
\end{figure}

\subsection{Analysis in the CPL case} 
The expression of FIM for a generic position of the dipole is too cumbersome to gain insights into the problem. The analysis becomes much easier if the following assumption is made.
\begin{assumption}[CPL assumption]\label{CPL assumption}
The center $C$ of the dipole is on the line perpendicular to  $\mathcal R_o$ passing through the point $O$, as shown in Fig.~\ref{fig:holo-surface}. 
\end{assumption}
Under Assumption~\ref{CPL assumption}, we have that $y_C=0$ and $z_C=0$ (but unknown), and the Fisher's information matrix becomes diagonal (e.g.,~\cite{Hu2018b}). The following result is found.
\begin{lemma}
Under Assumption~\ref{CPL assumption}, the CRBs for the estimation of $x_C$, $y_C$ and $z_C$, are given by
\begin{align}
\label{eq:crb_x}
\mathrm{CRB}(x_C)=[\mathbf{F}]^{-1}_{11}&= \dfrac{\mathrm{SNR}^{-1}}{\left(\mathscr{I}_1 +  \mathscr{I}_2\right)}\\ \label{eq:crb_y}
\mathrm{CRB}(y_C)=[\mathbf{F}]^{-1}_{22}&= \dfrac{\mathrm{SNR}^{-1}}{\left(\mathscr{I}_3 +  \mathscr{I}_4\right)}\\ \label{eq:crb_z}
\mathrm{CRB}(z_C)=[\mathbf{F}]^{-1}_{33}&= \dfrac{\mathrm{SNR}^{-1}}{\left(\mathscr{I}_5 +  \mathscr{I}_6\right)}	\end{align}
where $\mathrm{SNR} =|\chi|^2 / \sigma^2$ is the signal-to-noise ratio
and 
\begin{align}
\label{I1}
\mathscr{I}_1& \triangleq  k^2 x_C^2 \int\limits_{{\scriptscriptstyle-L/2}}^{{\scriptscriptstyle L/2}} \int\limits_{{\scriptscriptstyle -L/2}}^{{\scriptscriptstyle L/2}} \dfrac{x_C^2+y^2}{r^6}dydz\\
\label{I2}
\mathscr{I}_2& \triangleq \int\limits_{{\scriptscriptstyle-L/2}}^{{\scriptscriptstyle L/2}} \int\limits_{{\scriptscriptstyle -L/2}}^{{\scriptscriptstyle L/2}}\dfrac{x_C^4+x_C^2y^2-x_C^2z^2+y^2z^2+z^4}{r^{8}}dydz\\
\label{I3}
\mathscr{I}_3 &\triangleq k^2 \int\limits_{{\scriptscriptstyle-L/2}}^{{\scriptscriptstyle L/2}} \int\limits_{{\scriptscriptstyle -L/2}}^{{\scriptscriptstyle L/2}}\dfrac{y^2(x_C^2+y^2)}{r^6}dydz\\
\label{I4}
\mathscr{I}_4&\triangleq \int\limits_{{\scriptscriptstyle-L/2}}^{{\scriptscriptstyle L/2}} \int\limits_{{\scriptscriptstyle -L/2}}^{{\scriptscriptstyle L/2}} \dfrac{y^4+x_C^2y^2-y^2z^2+x_C^2z^2+z^4}{r^{8}}dydz\\
\label{I5}
\mathscr{I}_5&\triangleq k^2 \int\limits_{{\scriptscriptstyle-L/2}}^{{\scriptscriptstyle L/2}} \int\limits_{{\scriptscriptstyle -L/2}}^{{\scriptscriptstyle L/2}}\dfrac{z^2(x_C^2+y^2)}{r^6}dydz\\
\label{I6}
\mathscr{I}_6&\triangleq \int\limits_{{\scriptscriptstyle-L/2}}^{{\scriptscriptstyle L/2}} \int\limits_{{\scriptscriptstyle -L/2}}^{{\scriptscriptstyle L/2}}\dfrac{(x_C^2+y^2)\left(x_C^2+y^2+4z^2\right)}{r^{8}}dydz.
\end{align}
\end{lemma}
\begin{proof}
See~\cite[Appendix B]{damico2021}.
\end{proof}
\vspace{-0.3cm}
Although the CPL assumption results in a considerable simplification (as it makes FIM diagonal), the expressions \eqref{I1}--\eqref{I6} are still rather complicated. Further simplifications are provided in the following corollary in the regime $x_C\gg\lambda$.
\begin{corollary}If $x_C\gg\lambda$, then
\begin{align}
\label{eq:crb_x_approx}
\mathrm{CRB}(x_C)&\approx \dfrac{\mathrm{SNR}^{-1}}{\mathscr{I}_1}\\
\label{eq:crb_y_approx}
\mathrm{CRB}(y_C)&\approx \dfrac{\mathrm{SNR}^{-1}}{\mathscr{I}_3}
\end{align}
where $\mathscr{I}_1$ can be computed in closed-form
\begin{align}
\label{eq:I1}
\mathscr{I}_1=\dfrac{\rho k^2}{2(1+\rho^2)}\left[\dfrac{(7+6\rho^2)}{\sqrt{1+\rho^2}} \arctan \dfrac{\rho}{\sqrt{1+\rho^2}}+\dfrac{\rho}{(1+2\rho^2)}\right]
\end{align}
with $\rho \triangleq L/x_C$. As for $\mathscr{I}_3$, we have
\begin{equation}
\label{I3.1}
\mathscr{I}_3^{(l)} < \mathscr{I}_3 \le \mathscr{I}_3^{(u)}
\end{equation}
with
\begin{align}
\label{I3.low}
k^{-2}\mathscr{I}_3^{(l)}&=\dfrac{3\pi}{8} \ln (1+ \rho^2) - \dfrac{\pi}{16} \dfrac{\rho^2(5 \rho^2 +6)}{(1+\rho^2)^2}\\
\label{I3.up}
k^{-2}\mathscr{I}_3^{(u)}&= \dfrac{3\pi}{8} \ln (1+ 2 \rho^2) - \dfrac{\pi}{4} \dfrac{\rho^2(5 \rho^2 +3)}{(1+2\rho^2)^2}
\end{align}
Analogously, we have that
\begin{equation}
\label{I5.1}
\mathscr{I}_5^{(l)} < \mathscr{I}_5 \le \mathscr{I}_5^{(u)}
\end{equation}
with
\begin{align}
\label{I5.low}
k^{-2}\mathscr{I}_5^{(l)}&=\dfrac{\pi}{8} \ln (1+ \rho^2) + \dfrac{\pi}{16} \dfrac{\rho^2(\rho^2 -2)}{(1+\rho^2)^2}\\
\label{I5.up}
k^{-2}\mathscr{I}_5^{(u)}&= \dfrac{\pi}{8} \ln (1+ 2 \rho^2) + \dfrac{\pi}{4} \dfrac{\rho^2(\rho^2 -1)}{(1+2\rho^2)^2}
\end{align}
Finally, $\mathscr{I}_6$ is given by
\begin{align}
\label{I6.1}\notag
x_C^2\mathscr{I}_6 &=\dfrac{\rho(18 \rho^4+38 \rho^2 +17)}{4 (1+\rho^2)^{5/2}}\arctan\left(\frac{\rho}{\sqrt{\rho^2+1}}\right)\\
&+\dfrac{\rho^2(6 \rho^4+2 \rho^2 -1)}{4(1+\rho^2)^2(1+2\rho^2)^2}.
\end{align}
\end{corollary}
\begin{proof}
Can be derived by following the steps in~\cite[App. C]{damico2021}.
\end{proof}
\vspace{-0.3cm}
\subsection{Asymptotic analysis for CPL when $\rho= L/x_C \to \infty$}
The results of Corollary 1 allow a simple analysis of the behaviour of the CRBs \eqref{eq:crb_x}--\eqref{eq:crb_z} in the asymptotic regime $\rho \to \infty$. 
\begin{corollary}
Under Assumption~\ref{CPL assumption}, if $\rho = L/x_C \to \infty$ then
\begin{align}
\label{eq:CRBXtoINF}
\lim\limits_{\rho \to \infty} \mathrm{CRB}(x_C)&=\dfrac{\mathrm{SNR}^{-1}} {3 \pi^3} \lambda^2\\
\label{eq:CRBYtoINF}
\lim\limits_{\rho \to \infty} \mathrm{CRB}(y_C)&=\dfrac{\mathrm{SNR}^{-1}} {3 \pi^3} \dfrac{\lambda^2}{\ln \rho}\\ \label{eq:CRBZtoINF}
\lim\limits_{\rho \to \infty} \mathrm{CRB}(z_C)&=\dfrac{\mathrm{SNR}^{-1}} { \pi^3} \dfrac{\lambda^2}{\ln \rho}.
\end{align}
\begin{proof}
See~\cite[App. D]{damico2021}.
\end{proof}
\end{corollary}
It is interesting to note that the CRBs for the estimation of $y_C$ and $z_C$ goes to zero as $\rho$ increases unboundedly. This is in contrast with the results in~\cite[Eq. (26)]{Hu2018b} where it is shown that the asymptotic CRBs are identical for all the three dimensions and depend solely on the wavelength $\lambda$. This difference is a direct consequence of the different radiation and signal models used for the computation of CRBs. Indeed, in \cite{Hu2018b} the source is assumed to radiate isotropically and the scalar field \eqref{PVScalarField} is used for deriving the bounds.

%% file: SectionIV.tex
\section{Numerical Analysis}

\begin{figure}[t!]
	\includegraphics[width=1.1\columnwidth]{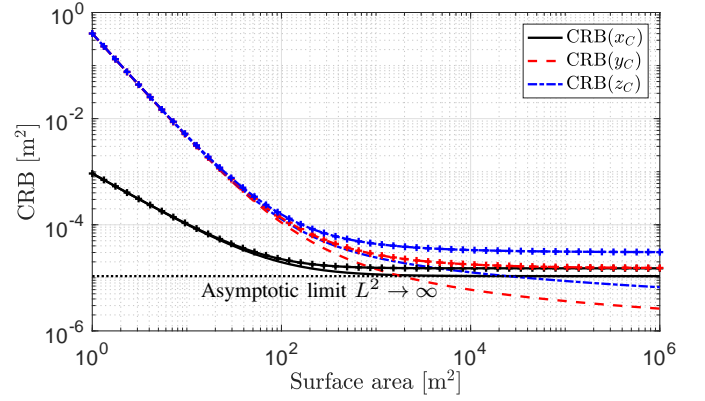}
	\put(-200,38){\footnotesize Asymptotic limit $L^2\to \infty$}
	\caption{CRB for the three Cartesian coordinates in m$^2$ for the CPL condition as function of surface area when $\lambda =0.01$\,m and $x_C = 6$\,m. The curves with markers indicate the performance obtained by using~\eqref{eq:electric-field-lund}.}
	\label{fig:crb-vs-size}
\end{figure}

Numerical simulations are now used to evaluate the estimation accuracy under different operating conditions. 
We assume that $|\chi|^2 = 1$\,V$^2$, $\sigma^2=10$\,V$^2$ such that $\mathrm{SNR}=-10$\,dB. Due to space limitations, only the CPL case is considered. However, the general conclusions and behaviours are also valid for other cases; an accurate analysis is provided in the extended version~\cite{damico2021}.

Fig.~\ref{fig:crb-vs-size} plots the CRB separately for the three Cartesian coordinates in m$^2$ as function of the surface area $L^2$ when $\lambda =0.01$\,m (corresponding to $f_c=30$\,GHz) and the dipole is located at a distance of $x_C = 6$\,m. We see that all the CRBs decrease fast with the surface area. An accuracy on the order of tens of centimeters (as required for example in future automotive and industrial applications, e.g.,~\cite{Witrisal5Gaccuracy}) is achieved for surface areas of practical interest, i.e., in the range $1 \le L^2\le 25$\,m$^2$.  We notice that, in this range, $\mathrm{CRB}(x_C)$ is much lower than $\mathrm{CRB}(y_C)$ and $\mathrm{CRB}(z_C)$. However, as $L^2$ increases, $\mathrm{CRB}(x_C)$ converges to the lower limit~\eqref{eq:CRBXtoINF} whereas the other two decrease unboundedly as $1/\ln \rho$. These results corroborate Corollary 2 and show that $L^2>10^4$\,m$^2$ is needed for $\mathrm{CRB}(x_C)$ to approach the lower limit. Comparisons are made with the CRBs obtained from~\eqref{PVScalarField} in Remark 1 (see the curves with markers). Particularly, notice that, under Assumption 3, ~\eqref{PVScalarField} reduces to
\begin{align}\label{eq:electric-field-lund}
E=\chi \dfrac{e^{- \imagunit kr}}{r^{5/2}} \sqrt{x_C\left[ x_C^2+(y-y_C^2) \right]}.
\end{align}
Only marginal differences are observed between the two methods for areas of practical interest. However, different limits are achieved as $L^2$ increases. In conclusion, both are accurate and might be used to predict scaling behaviors, but the proposed one is needed to study the fundamental limits.


Fig.~\ref{fig:distance_comp} plots CRBs as a function of the distance $x_C$ of the dipole when $L^2 = 9$\,m$^2$ and $\lambda =0.01$ or $0.001$\,m. The CRBs have the same behavior irrespective of the wavelength $\lambda$. As expected, higher accuracies are achieved when $\lambda = 0.001$\,m. We see that $\mathrm{CRB}(y_C)$ and $\mathrm{CRB}(z_C)$ increase fast as the distance increases. On the other hand, $\mathrm{CRB}(x_C)$ starts to increase when the surface has a size comparable to the distance. In the considered scenario, this happens for $x_C \ge L/3$. 

\begin{figure}[t!]
	\includegraphics[width=1.1\columnwidth]{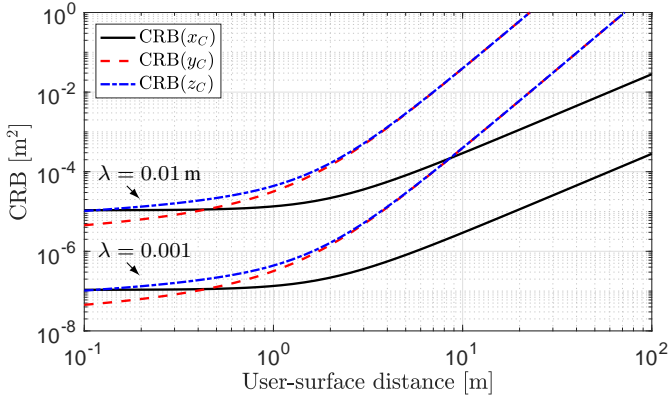}
	\put(-236,53){\footnotesize $\lambda =0.001$}
	 \put(-225,51){\vector(1, -1){5}}
	\put(-236,82){\footnotesize $\lambda =0.01$\,m}
	\put(-225,80){\vector(1, -1){5}}
	\caption{CRB in m$^2$ as a function of the dipole distance $x_C$ in CPL when $L^2 = 9$\,m$^2$ and $\lambda =0.01$ or $0.001$\,m.}
	\label{fig:distance_comp}
\end{figure}
